\documentclass[onecolumn,12pt]{IEEEtran}
\IEEEoverridecommandlockouts
\usepackage[russian,english]{babel}
\usepackage{cite}
\usepackage{amsmath,amssymb,amsfonts}
\usepackage{amsthm}
\usepackage{algorithmic}
\usepackage{graphicx}
\usepackage{textcomp}
\usepackage{xcolor}
\usepackage{tikz, color, soul}
\usepackage{multirow}
\newtheorem{theorem}{Theorem}

\newtheorem{remark}{Remark}
\newtheorem{corollary}{Corollary}
\usetikzlibrary{shapes.geometric,arrows,positioning}
\tikzstyle{decision} = [diamond, draw, fill=blue!20, 
    text width=5em, text badly centered, node distance=4cm, inner sep=0pt]
\tikzstyle{block} = [rectangle, draw, fill=blue!20,  text centered, rounded corners, minimum height=4em]
\tikzstyle{line} = [draw, -latex']
\tikzstyle{cloud} = [draw, ellipse,fill=red!20, node distance=6.6cm,
    minimum height=1em]
\tikzstyle{algorithm} = [rectangle, draw, fill=green!20,  text centered, rounded corners, minimum height=4em, minimum width =6em]
\tikzstyle{initialization} = [rectangle, draw,   text centered, minimum height=4em, minimum width =6em]
\def\BibTeX{{\rm B\kern-.05em{\sc i\kern-.025em b}\kern-.08em
    T\kern-.1667em\lower.7ex\hbox{E}\kern-.125emX}}
    
    \tikzstyle{block}=[draw, rectangle, minimum height=1cm, text width=1.5cm, text centered, draw=darkgray, font=\small]
\tikzstyle{block_medium}=[draw, rectangle, minimum height=1.5cm, text width=2cm, text centered, draw=darkgray, font=\small]
\tikzstyle{block_large}=[draw, rectangle, minimum height=2cm, text width=2cm, text centered, draw=darkgray, font=\small]
\tikzstyle{line} = [draw, -latex]

\newtheorem{Example}{Example}

\newcommand{\y}{\mathbf{y}}
\newcommand{\x}{\mathbf{x}}

\renewcommand{\a}{\mathbf{a}}
\renewcommand{\u}{\mathbf{u}}
\renewcommand{\v}{\mathbf{v}}

\title{Correcting one error in channels with feedback
\thanks{
Ilya Vorobyev and Christian Deppe are with Institute of Communications Engineering, 
Technical University of Munich, 
Munich, Germany. 
(email: ilya.vorobyev@tum.ru, christian.deppe@tum.de)

Alexey Lebedev and Vladimir Lebedev are with Kharkevich Institute for Information Transmission Problems, Moscow, Russia. (email: al\_lebed95@mail.ru, lebedev37@mail.ru)}
}

\author{Ilya Vorobyev, Alexey Lebedev, Vladimir Lebedev, Christian Deppe
}

\begin{document}

\maketitle

\begin{abstract}
We address the problem of correcting a single error in an arbitrary discrete
memoryless channel with error-free instantaneous feedback. For the case of a one-time
feedback, we propose a method for constructing optimal transmission strategies. The
obtained result allows us to prove that for a binary channel, two feedbacks are
sufficient to transmit the same number of messages as in the case of complete
feedback. We also apply the developed techniques to a~binary asymmetric channel to
construct transmission strategies for small lengths.

\end{abstract}

\section{Introduction}

We analyze the problem of correcting a single error in an arbitrary discrete
memoryless channel with instantaneous error-free feedback. In what follows we always
assume all these conditions---memoryless channel and error-free instantaneous
feedback---to be fulfilled. The most attention is paid to binary symmetric and
asymmetric channels. In a binary symmetric channel, any symbol can be transmitted
erroneously, for example, 0 instead of 1, or vice versa. The word \emph{symmetric\/}
is usually omitted, and such a channel is simply referred to as a binary channel. In
a binary asymmetric channel, 0 can be received instead of a transmitted symbol 1, but
the symbol 0 is always transmitted without errors. We consider a combinatorial model
of such a channel with feedback and a single transmission error.

It is known that the problem of correcting $t$ errors in a binary channel with
complete feedback is equivalent to the following combinatorial search problem. It is
required to find an element $x\in\mathcal{M}$ using $n$ questions of the following
type: ``Does an element $x$ belong to a subset~$A$ of a~set~$\mathcal{M}$?''
Questions are asked in succession, i.e., each next question may depend on answers to
the preceding ones. The opponent who answers the questions knows $x$ and is allowed
to lie at most $t$ times. This problem was first formulated by
R\'enyi~\cite{renyi61}. For a linear number of errors in a binary channel with
complete feedback, the optimal transmission rate was computed by Berlekamp~\cite{B68}
and Zigangirov~\cite{zigangirov1976}. This problem became popular after Ulam in his
biography \cite{ulam1991adventures} asked a similar question for $M=10^6$. Optimal
strategies for all $M$ have been found in \cite{P87} for $t=1$, in \cite{G90} for
$t=2$, and in \cite{D00} for $t=3$. Tables of optimal strategies for various values
of $t$ and $M\le 2^{20}$ are presented~in~\cite{D02}.

Error correction in a binary asymmetric channel with complete feedback is equivalent
to a~version of Ulam's problem with halflie first described in
\cite{rivest1980coping}. The difference from the original problem is that lying is
allowed only when the true answer is affirmative. A good survey of results on this
problem can be found in \cite{Cicalese13}. For a fixed number $t$ of errors, the
maximum cardinality of $\mathcal{M}$ is asymptotically equivalent to
$2^{n+t}{\bigm/}\smash[b]{\binom{n}{t}}$. For $t=1$, this was proved
in~\cite{cicalese2000optimal}, and for an arbitrary~$t$,
in~\cite{dumitriu2004halfliar,spencer2003halflie}.\strut

Note that for a fixed number of errors, even one-time feedback is sufficient to
transmit asymptotically the same number of messages as in the case of complete
feedback. For a nonbinary symmetric channel, this was proved in \cite{bassalygo2005},
and for an arbitrary discrete channel, in \cite{dumitriu2005two}.

A key result of the present paper is the description of optimal strategies with
one-time feedback and a single error for an arbitrary discrete channel. The developed
technique is applied to construct single-error-correcting transmission strategies for
a binary channel with one- or two-time feedback, and also to construct
single-error-correcting strategies in a binary asymmetric channel with one-time
feedback. The most interesting of the obtained results, in our opinion, is
constructing a~strategy with two feedbacks that corrects a single error in a binary
channel and transmits as many messages as a completely adaptive strategy.

The rest of the paper is organized as follows. In Section~\ref{sec::definitions} we
give basic definitions. In Section~\ref{sec::one feedback} we formulate and prove a
theorem describing the structure of optimal strategies with a single error and
one-time feedback. In Section~\ref{sec::bsc with feedback} the main theorem is
applied to construct a single-error-correcting transmission strategy in a binary
channel with two feedbacks that allows to transmit as many messages as in the case of
complete feedback. In the last section, the developed technique is applied to find
good strategies for a binary asymmetric channel with a single error and one-time
feedback.

\section{Basic Definitions}\label{sec::definitions}

Consider a channel with $q$-ary input alphabet $\mathcal{X}=\{0,\ldots, q-1\}$ and
output alphabet $\mathcal{Y}=\mathcal{X}$. The encoder transmits a message
$\x\in\mathcal{X}^n$, and the decoder receives a message
$\y\in\mathcal{Y}^n$. The prefix of length $p$ of a vector $\y$ will be
denoted by $\y_{\overline{p}}$. By an error we mean replacing a symbol~$q_1$ of
a sequence~$\x$ by a symbol $q_2$, $q_1\ne q_2$. We define a bipartite graph $G$
with the left-hand part corresponding to elements from $\mathcal{X}$, and the
right-hand part, to elements from $\mathcal{Y}$. We connect $q_1\in\mathcal{X}$ and
$q_2\in\mathcal{Y}$, $q_1\ne q_2$, by an edge if an error may change the symbol $q_1$
to $q_2$. An example of such a graph for a one-way ternary channel is shown in
Fig.~\ref{fig2}.

\begin{figure}[h]
\centering\vskip3pt
\begin{tikzpicture}
\node (A) at (0,0) {0}; \node (B) at (2,0) {0}; \node (C) at (0,-2) {1}; \node (D) at
(2,-2) {1}; \node (E) at (0,-4) {2}; \node (F) at (2,-4) {2}; \path[line] (C) -- (B);
\path[line] (E) -- (B); \path[line] (E) -- (D);
\end{tikzpicture}
\caption{Error graph for a one-way ternary channel.}\label{fig2}
\bigskip
\end{figure}
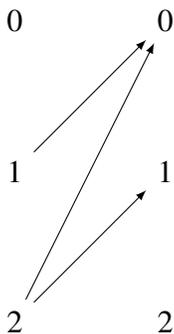

In this paper we consider transmission over a channel with $k$ feedbacks. Let the
codeword length~$n$ be divided into $k+1$ parts:
$$
n=n_1+n_2+\ldots+n_{k+1}.
$$
The encoder transmits a message $m\in [M]$. The first $n_1$ transmitted symbols
$x_1,\ldots, x_{n_1}$ depend on the message $m$ only. After $N_{i-1}:=n_1+\ldots
+n_{i-1}$, $i\ge 2$, symbols are transmitted, the encoder has values of the received
symbols $\y^{}_{\overline{N_{i-1}\!}}$ from the feedback channel. The encoder
sends the $i$th block of $n_i$ symbols, which is a function of the message $m$ and of
the symbols $\y^{}_{\overline{N_{i-1}\!}}$ received by the encoder. The case of
$k=0$ corresponds to a channel without feedback, and the case of $k=n-1$, to a
channel with complete feedback.

We define the cloud $B_t(m)$ (or $B(m)$ for $t=1$) for a message $m$ to be the set of
sequences $\y$ that can be obtained at the output of the channel with at
most~$t$ errors during the transmission of this message. We refer to the collection
of disjoint clouds $B_t(m)$, $m\in [M]$, as a $t$-error-correcting
code~$\mathcal{C}$. Points of the space that do not belong to any cloud will be
referred to as free points and will be denoted by~$\mathcal{F}(\mathcal{C})$. Codes
that do not use feedback will be called nonadaptive.

Note that for a symmetric channel without feedback, clouds are spheres of radius $t$
in the Hamming metric. For a symmetric channel, sizes of all clouds are the same, but
for an arbitrary error graph this is not the case. For constructions proposed in the
present paper, it makes sense to find codes with the maximum number of free points
for each length and each cardinality. Such codes will be called $F$-optimal.

As an example, let us describe the structure of clouds for a binary channel with a
single error and complete feedback. Every cloud $B(m)$ contains a sequence $\y$
which will be transmitted if there are no errors in the channel. We call it a root
sequence. For any coordinate $i$, the cloud contains a~sequence $\y(i)$ which
coincides with $\y$ in the first $i-1$ positions, differs from it in the $i$th
position, and has arbitrary symbols in all other positions. Hence it is seen that
each cloud consists of at~least $n+1$ sequences. In particular, this yields the
Hamming bound on the maximum number of transmitted messages.

\section{One-Time Feedback}\label{sec::one feedback}

In this section we propose a transmission strategy for the case of a single error and
one-time feedback. We divide the codeword length $n$ into two parts, $n_1$ and $n_2$,
with $n=n_1+n_2$. We define a bipartite graph $H=(U\sqcup V, E)$ as follows. The
left- and right-hand parts consist of $q^{n_1}$ vertices corresponding to the sets of
input and output sequences. Vertices $u$ and $v$ are connected by an edge if the
sequence corresponding to $v$ can be obtained from the sequence corresponding to~$u$
as a~result of a single error. Note that we do not connect vertices corresponding to
identical sequences (this corresponds to the case of no error).

\begin{theorem}\label{th::main}
Let a graph $H=(U\sqcup V, E)$ be given. A strategy allowing to transmit
\begin{equation}\label{number of transmitted words}
M=\sum\limits_{u\in U}M(u)
\end{equation}
messages exists if and only if there exists a family of single-error-correcting codes
$C(u)$ of length~$n_2$ and cardinality $M(u)$ with $F(u)$ free points that satisfy
the condition
\begin{equation}\label{constraints}
\sum\limits_{u:\: (u,v)\in E} M(u)\le F(v)
\end{equation}
for any $v\in V$.
\end{theorem}

\begin{proof}
Let us describe an arbitrary coding strategy. First, a sequence $\u$ of length
$n_1$ is transmitted, which corresponds to a vertex $u$ in the left-hand part $U$ of
the graph $H$. Let the number of messages such that transmitting them begins with the
sequence $\u$ be $M(u)$. Consider the case of no error in the first $n_1$
symbols. On the remaining $n_2$ symbols, we need to transmit $M(u)$ different
messages, and one error may occur. Therefore, we need to use a
single-error-correcting code $\mathcal{C}(u)$ of length $n_2$ and cardinality $M(u)$.
Denote by $F(u)$ the number of free points of $C(u)$.

If there was an error in the first $n_1$ symbols and instead of $\u$ a sequence
$\v$ was received, to transmit $M(u)$ messages that start with $\u$ we need
$M(u)$ points, which must be free points of $\mathcal{C}(v)$.

Thus, for each message $\v$ there should exist a code $\mathcal{C}(v)$ with free
points distributed among the sequences $\u$ from which the sequence $\v$
can be reached, and every such sequence~$\u$ must get at least $M(u)$ free
points, which is possible if and only if condition \eqref{constraints} is satisfied.

Now we describe the decoding algorithm. Let a sequence of the first~$n_1$ received
symbols correspond to a vertex $v\in V$; denote the sequence of the last $n_2$
symbols by $\a$. If $\a$ is not a free point of~$\mathcal{C}(v)$, this
means that an error occurred in the second part of the message. In this case, the
second part of the message corresponds to the center of the sphere to which the point
$\a$ belongs.{\parfillskip=0pt\par}

If the sequence $\a$ is a free point of $\mathcal{C}(v)$, then it corresponds to
some sequence $\u$ from which $\v$ can be obtained. Then precisely this
sequence $\u$ has been transmitted at the first encoding stage. The second part
of the message is recovered based on which point $\a$ out of at least $M(u)$
points corresponding to $\u$ was used.
\end{proof}

We are not aware of any efficient (polynomial in the code length) method to find an
optimal family of codes satisfying \eqref{constraints}. However, even choosing
identical codes for all sequences $\u\in\mathcal{X}^{n_1}$ can provide a good
result, as is shown in Corollary~\ref{cor::BSC with one feedback}.

For an example showing that choosing identical codes for a binary channel is not
optimal, consider the case of $n_1=2$ and $n_2=1$. When identical codes are chosen,
the maximum number of messages is always divisible by $2^{n_1}$, and in this case it
is~0, since even adaptively, no more than two messages can be transmitted on length
3. If we choose two different codes, we can transmit two messages.

The following statement for a binary symmetric channel will be used below to
construct an optimal strategy with two feedbacks.

\begin{corollary}\label{cor::BSC with one feedback}
Let $n_2=2^k-1$, $n_1=n-n_2$, $k\ge 1$. Then in a symmetric channel with a
single error and one-time feedback we can transmit
$$
M_1(n)=2^{n_1} \left\lfloor \frac{2^{n_2}}{n_1+n_2+1} \right\rfloor
$$
messages.
\end{corollary}

\begin{proof}
To each point $\u$, assign as $\mathcal{C}(u)$ the Hamming code of length $n_2$
with $x$ codewords deleted. We choose $x$ so that to satisfy the constraints
\eqref{constraints}. This is equivalent to the inequality
$$
x2^k\ge n_1 (2^{n_2-k}-x),
$$
whence we obtain
$$
x\ge\frac{n_12^{n_2-k}}{n_1+2^k}.
$$
Then we may take $x=\Bigl\lceil\frac{n_12^{n_2-k}}{n_1+2^k}\Bigr\rceil$. The number
of remaining words in the chosen codes is
$$
2^{n_2-k}-\left\lceil \frac{n_12^{n_2-k}}{n_1+2^k}\right\rceil= \left\lfloor
\frac{2^{n_2}}{n_1+2^k} \right\rfloor
=
\left\lfloor \frac{2^{n_2}}{n_1+n_2+1} \right\rfloor
$$

The total number of transmitted messages is
$$
M_1(n)=2^{n_1} \left\lfloor \frac{2^{n_2}}{n_1+n_2+1} \right\rfloor.
$$
\end{proof}

\section{Binary Symmetric Channel with Feedback}\label{sec::bsc with feedback}

Denote by $Alg_k(n)$ a transmission strategy in a channel of length~$n$ with a
single error and $k$ feedbacks. Now we describe an algorithm to construct a strategy
$Alg_k(n)$ given $Alg_{k-1}(n-1)$, which will be used in what follows.

\textbf{DADA (Double and Delete Algorithm) algorithm for constructing a strategy
\kern-2.1pt$Alg_k(n)$ given $Alg_{k-1}(n-1)$.}

Recall that every cloud in a binary symmetric channel of length $n-1$ contains a root
message and $n-1$ additional messages that coincide with the root in the first $i-1$
symbols and differ from it in the $i$th symbol, $i=1,2,\ldots, n-1$. From each cloud
of messages of~length $n-1$, we construct two sets of messages of length~$n$ by
adding to each message the prefix~0 for the first set and 1 for the second. To make a
cloud on length $n$ from the first (second) set, it~suffices to add any message
beginning with 1 (0). We will refer to such sets as incomplete clouds. Next, from
each free point we make two free points by adding a prefix 0 or 1. We use up all
available free sequences to turn some number of incomplete clouds into clouds. If the
number of incomplete clouds is not greater than the number of free sequences of
length $n$, by the end of this procedure we will have $2M(n-1)$ clouds and some
number of free points, where $M(n-1)$ is the number of messages transmitted by the
$Alg_{k-1}(n-1)$ algorithm. In this case, the DADA algorithm is completed.

Otherwise, by the end of this procedure we will have some number of clouds and some
number of incomplete clouds completely covering the space.

After that, we will take one incomplete cloud of sequences beginning with 1 and one
incomplete cloud of sequences beginning with 0 and eliminate them by turning all
their elements into free points. This operation yields $2n$ free points. Then free
points are used to turn incomplete clouds into clouds. The operation is repeated
until the incomplete clouds are over.

At the end of the procedure, there remain an even number of clouds and at most $2n$
free points. If the number of free points is $2n$, this means that these free points
have been obtained just now from two incomplete clouds. Then we reconstruct one of
these incomplete clouds back and turn it into a cloud by adding one free point. As a
result, we obtain an additional cloud.

Thus, we have proved the following.

\begin{theorem}\label{th::plus1construction}
Assume that on length $n-1$ we have constructed $M(n-1)$ clouds for transmitting
messages with a single error and  $k-1$ feedbacks, $k=1,\ldots, n-1$. Let
$$
U(n)=2\left\lfloor\frac{2^n}{2(n+1)}\right\rfloor,\qquad r(n)=2^n-(n+1)U(n).
$$
Then the DADA algorithm constructs a strategy\/ $Alg_k(n)$ transmitting $M(n)$
messages, where
$$
M(n)=
\begin{cases}
2M(n-1) & \text{if\/}\ 2M(n-1)\le\frac{2^n}{n+1},\\[5pt] U(n) & \text{if\/}\
2M(n-1)> \frac{2^n}{n+1}\ \text{and\/}\ r(n)<2n,\\[5pt] U(n)+1 & \text{if\/}\
2M(n-1)> \frac{2^n}{n+1}\ \text{and\/}\ r(n)\ge 2n.
\end{cases}
$$
\end{theorem}

In the complete feedback case, the optimum number of messages that can be transmitted
with a~single error has been computed in~\cite{P87}. In
Theorem~\ref{th::simpleProofCompleteFeedback} we present a new simpler proof of this
result.{\parfillskip=0pt\par}

\begin{theorem}\label{th::simpleProofCompleteFeedback}
Let
$$
U(n)=2\left\lfloor\frac{2^n}{2(n+1)}\right\rfloor,\qquad r(n)=2^n-(n+1)U(n).
$$
Then one can transmit $M_{\rm ad}(n)$ messages through a channel with a single
error, where
\begin{equation}\label{binary complete feedback}
M_{\rm ad}(n)=
\begin{cases}
U(n) & \text{if\/}\ r(n)<2n,\\ U(n)+1 & \text{if\/}\ r(n)\ge 2n.
\end{cases}
\end{equation}
Moreover, this number of messages is optimal.
\end{theorem}

\begin{remark}
In fact, $r$ is always even and is less than $2n+2$; therefore, the condition $r\ge
2n$ in the last line of \eqref{binary complete feedback} can be replaced with $r=2n$.
Note that the second case is realized very rarely. Namely, the code cardinality is
$U(n)+1$ for $n=1,2$, and the next length for which this happens is $49\,736$. Thus,
the optimum number of messages is most often the largest even number that does not
exceed the Hamming bound.
\end{remark}

\begin{proof}
We will construct a transmission strategy inductively. For $n\le 8$, the formula can
be verified by hand.

Now assume that for length $n-1$, $n\ge 9$, we have constructed $M_{\rm ad}(n-1)$
clouds. Note that the number of incomplete clouds is at least
$$
2M_{\rm ad}(n-1)\ge\frac{2^n}{n}-4> \frac{2^n}{n+1}
$$
for $n\ge 9$. We use the DADA algorithm to construct an adaptive strategy on length
$n$ from the strategy on length $n-1$. Since
$$
2M_{\rm ad}(n-1)>\frac{2^n}{n+1},
$$
we conclude that $M_{\rm ad}(n)$ equals either $U(n)$ or $U(n)+1$ depending on
$r(n)$, as required.
\end{proof}

\begin{theorem}
In a binary channel with two feedbacks and a single error, $M_{\rm ad}(n)$ messages
can be transmitted, i.e., the same number as for the transmission with
complete feedback.
\end{theorem}

Note that one-time feedback is not sufficient for that, which is seen from
Table~\ref{table::sym 1fb and ad}.

\begin{table}[h]
\centering\vskip3pt
\begin{tabular}{|l|l|l|l|l|l|l|l|l|l|l|l|l|l|l|}
\hline
$n$      & 3 & 4 & 5 & 6 & 7  & 8  & 9  & 10 & 11  & 12  & 13  & 14   & 15   & 16   \\ \hline
$M_1$    & 2 & 2 & 4 & 8 & 16 & 28 & 50 & 90 & 168 & 312 & 580 & 1088 & 2048 & 3854 \\ \hline
$M_{ad}$ & 2 & 2 & 4 & 8 & 16 & 28 & 50 & 92 & 170 & 314 & 584 & 1092 & 2048 & 3854 \\ \hline
\end{tabular}
\caption{Maximum numbers transmitted messages in a binary symmetric channel with
a single error, with one-time feedback and with complete feedback.}\label{table::sym 1fb and ad}
\end{table}

\begin{proof}
We use the DADA algorithm to construct $Alg_2(n)$ given $Alg_1(n-1)$. For
$Alg_1(n-1)$ we take the strategy constructed in Corollary~\ref{cor::BSC with one
feedback} with
$$
M_1(n-1)=\biggl\lfloor\frac{2^{n_2}}{n_1+n_2+1}\biggr\rfloor 2^{n_1},
$$
where $n_2=2^k-1$ and $n_1=n-1-n_2$.

Notice (see Table~\ref{table::sym 1fb and ad}) that for $n\le 9$ even one-time
feedback is sufficient to transmit the same number of messages as for coding with
complete feedback. Therefore, it suffices to prove the statement for $n\ge 10$. Let
us show that
$$
2M_1(n-1)> \frac{2^n}{n+1}
$$
for $n\ge 10$.

This is equivalent to the inequality
$$
2^{n_1+1}\left\lfloor\frac{2^{n_2}}{n_1+n_2+1}\right\rfloor>
\frac{2^{n_1+n_2+1}}{n_1+n_2+2}.
$$

Reducing by $2^{n_1+1}$ and using the inequality $\lfloor x\rfloor> x-1$, we obtain
$$
\left\lfloor\frac{2^{n_2}}{n_1+n_2+1}\right\rfloor
>\frac{2^{n_2}}{n_1+n_2+1}-1\ge\frac{2^{n_2}}{n_1+n_2+2}.
$$
The latter inequality is equivalent to
$$
2^{n_2}\ge (n_1+n_2+1)(n_1+n_2+2).
$$
Recalling that $n_2\ge n_1$ and $n_2=2^k-1$, we conclude that the inequality holds
for $n_2\ge 15$.

Thus, we have proved the inequality $2M_1(n-1)> \frac{2^n}{n+1}$ for $n\ge 16$. For
$n\in[10, 15]$, the inequality can be checked by hand using Table~\ref{table::sym 1fb
and ad}. Applying Theorems~\ref{th::plus1construction}
and~\ref{th::simpleProofCompleteFeedback} yields the desired result.
\end{proof}

\section{Binary Asymmetric Channel}

In this section we apply the theorems obtained above to the binary asymmetric
channel. To this end, we need to compose tables of codes with many free points. To
find such codes, we use the linear programming method.

In \cite{delsarte1981bounds,klove1981upper,weber1987new}, linear programming was used
to prove upper bounds on the cardinality of nonadaptive codes correcting asymmetric
errors. We modify the methods from those papers to obtain upper bounds on the number
of free points in a code of fixed length and cardinality.

Denote by $M_Z(n,t)$ the maximum cardinality of a $t$-error-correcting asymmetric
code of length~$n$. Also, denote by $L(n,d,w)$ and $U(n,d,w)$ the lower and upper
bounds on the cardinality of a~con\-stant-weight code with weight $w$, length $n$,
and distance $d$.

\begin{theorem}\label{th::linear programming upper bounds}
Let\/ $n\ge 2t\ge 2$, $1\le M\le M_Z(n, t)$. Define
$$
\overline{F}(n, M, t)=\max\Biggl(2^n-\sum\limits_{i=0}^n\Biggl(
z_i\sum\limits_{j=0}^t\binom{i}{i-j}\Biggr)\Biggr),
$$
where the maximum is over all\/ $z_i$ satisfying the following conditions:
\begin{enumerate}
\item
$z_i$ are nonnegative integers;
\item
$z_0=1$, $z_1=z_2=\ldots= z_t=0$;
\item
$
\displaystyle\sum\limits_{i=1}^s\binom{n-w+i}{i}z_{w-i}
+\sum\limits_{j=0}^{t-s}\binom{w+j}{j}z_{w+j}\le\binom{n}{w}
$
for\/ $0\le s\le t<w<n-t$;
\item
$
\displaystyle\sum\limits_{j=s}^rz_jL(r-s, 2t+2, r-j)\le U(n+r-s, 2t+2, r)
$
for\/ $0\le s\le r$;
\item
$
\displaystyle\sum\limits_{j=s}^rz_{n-j}L(r-s, 2t+2, r-j)\le U(n+r-s, 2t+2, r)
$
for\/ $0\le s\le r$;
\item{\sloppy
$
\displaystyle\sum\limits_{i=1}^s\binom{n-w+i}{i}z_{w-i}
+\sum\limits_{j=0}^{t-s}\binom{w+j}{j}z_{w+j}
+\biggl(\binom{w+t-s+1}{w}-\binom{t+1}{t-s+1}
\left\lfloor\frac{w+t-s+1}{t+1}\right\rfloor\biggr)\linebreak \strut\times
z_{w+t-s+1}\le\binom{n}{w}
$
for\/ $0\le s\le t<w<n-t$,\\[2pt]
$
\displaystyle\sum\limits_{i=1}^s\binom{n-w+i}{i}z_{w-i}
+\sum\limits_{j=0}^{t-s}\binom{w+j}{j}z_{w+j}
+\biggl(\binom{n-w+s+1}{s+1}-\binom{t+1}{t-s}
\left\lfloor\frac{n-w+s+1}{t+1}\right\rfloor\biggr)\linebreak \strut\times
z_{w-s-1}\le\binom{n}{w}
$
for\/ $0\le s\le t<w<n-t$;}
\item
$
\displaystyle\sum\limits_{i=0}^nz_i=M.
$
\end{enumerate}
Then the number $F$ of free points in a code of length $n$ and cardinality $M$
correcting~$t$~asymmetric errors is not greater than $\overline{F}(n, M, t)$.
\end{theorem}

\begin{proof}
Denote by $z_i$, $0\le i\le n$, the number of codewords of weight $i$ in a code of
length~$n$ correcting $t$ asymmetric errors.
In~\cite{delsarte1981bounds,klove1981upper,weber1987new} it was proved that the $z_i$
must satisfy conditions~1 and~3\nobreakdash--6. It is easily seen that a code with
the maximum number of free points must satisfy condition 2. The last condition fixes
the cardinality of a considered code. The maximized expression $\overline{F}(n,M,t)$
corresponds to the number of free points in a code with weight distribution
$\{z_i\}$.
\end{proof}

We also apply the linear programming method to find codes with the maximum number of
free points. Fix a code length $n$, cardinality~$M$, and the number $t$ of
correctable asymmetric errors. Introduce $2^n$ binary variables $x_i$ corresponding
to all possible codewords. For each point~$p$, define the set $D_t(p)$ of codewords
from which this point can be reached as a result of $t$ asymmetric errors. Impose the
constraint $\sum\limits_{i\in D_t(p)}x_i\le 1$. For $t=1$, we will maximize the
number $2^n-\sum\limits_{i=0}^n z_i(i+1)$ of free points, where $z_i$ is the number
of codewords of weight $i$. Note that the number of free points can be expressed
through the variables $x_i$. Add the constraint $\sum x_i=M$ to fix the code
cardinality. Note that any solution to the linear programming problem (if exists)
yields an optimum number of free points for fixed code length and cardinality. To
speed up computations, we have added the constraints from Theorem~\ref{th::linear
programming upper bounds}.

Despite these optimizations, the program operates with $2^n$ variables, so a solution
can be found for small enough values of $n$ only. In Table~\ref{table::optimal number
of free points} we present the parameters of some $F$-optimal codes for $t=1$ and
$n=6$, $7$, $8$, and $9$. Optimal weight distributions are presented in
Table~\ref{table::optimal weight distributions}.

The parameters of the codes of length $n=6$, $8$, and $9$ coincide with the upper
bounds given by Theorem~\ref{th::linear programming upper bounds}. For $n=7$ and
$M=18$, we have obtained 48 free points instead of 49 given by the upper bound of
Theorem~\ref{th::linear programming upper bounds}; i.e., the upper bound of
Theorem~\ref{th::linear programming upper bounds} is not attained. All the other
values coincide with the upper bounds.

Codes with the optimal weight distribution for the lengths $n=7$ and $8$ have been
constructed in~\cite{delsarte1981bounds}. A code for $n=6$ was also previously known.
For the lengths $n=6$ and $8$ and all cardinalities $M\le M_Z(n, 1)$, optimal codes
can be obtained from the code of the maximum cardinality with the weight distribution
given in Table~\ref{table::optimal weight distributions} by deleting $M_Z(n, d)-M$
codewords of the maximum weight. For the length $n=7$ and cardinality $M=17$, we know
two codes with different weight distributions with the optimum number of free points:
$1+0+3+5+5+3+0+0$ and $1+0+3+5+6+1+1+0$. By~deleting codewords of the maximum weight
from the code with the second weight distribution, we obtain $F$-optimal codes for
all $M<17$. However, only the code with the first weight distribution can be
augmented to a code of cardinality 18.

The program works for all $n<9$. For larger lengths, the complexity is too high.
Since the $F$-optimal constructions for the lengths $n=6$ and $8$ are nested codes,
for the length $n=9$ we also restrict our search to such families. This approach
allowed us to find a family of nested codes such that the maximal code of cardinality
62 has the weight distribution presented in Table~\ref{table::optimal weight
distributions}. The number of free points in codes of this family coincides with the
upper bounds of Theorem~\ref{th::linear programming upper bounds} for all
cardinalities $M$. This means that the codes of the constructed family are
$F$-optimal. Note that the code of the maximum cardinality and of length $n=9$
constructed in~\cite{delsarte1981bounds} has 171 free points, whereas in our code
there are 177 free points.

\begin{table}[h]
\renewcommand{\arraystretch}{1.2}% for the vertical padding
\centering
\begin{tabular}{|c|c|c|c|c|c|c|}
\hline
\multirow{2}{*}{$n=6$}  & Cardinality $M$      & 12  & 11  & 10  & 9   & 8   \\ \cline{2-7} 
                        & Free points $F$ & 16  & 23  & 28  & 33  & 38  \\ \hline
\multirow{2}{*}{$n=7$}  & Cardinality $M$       & 18  & 17  & 16  & 15  & 14  \\ \cline{2-7} 
                        & Free points $F$ & 48  & 56  & 62  & 68  & 73  \\ \hline
\multirow{2}{*}{$n=8$}  & Cardinality $M$       & 36  & 35  & 34  & 33  & 32  \\ \cline{2-7} 
                        & Free points $F$ & 76  & 85  & 92  & 99  & 106 \\ \hline
\multirow{2}{*}{$n=9$}  & Cardinality $M$       & 62  & 61  & 60  & 59  & 58  \\ \cline{2-7} 
                        & Free points $F$ & 177 & 186 & 193 & 200 & 207 \\ \hline
%\multirow{2}{*}{$n=10$} & size        & 117 & 116 & 115 & 114 & 113 \\ \cline{2-7} 
                        %& free points & 327 & 338 & 347 & 356 & 365 \\ \hline
\end{tabular}
\bigskip
\caption{Optimal number of free points for $(n,M,1)$ codes.}\label{table::optimal number of free points}
\end{table}

\begin{table}[h]
\renewcommand{\arraystretch}{1.2}% for the vertical padding
\centering
\begin{tabular}{|c|c|}
\hline
Length and cardinality & \multicolumn{1}{c|}{Weight Distribution} \\ \hline
$n=6, M=12$      & 1+0+3+4+3+0+1                                             \\ \hline
$n=7, M=18$      & 1+0+3+5+5+3+1+0                                         \\ \hline
$n=7, M=17$      & 1+0+3+5+6+1+1+0                                           \\ \hline
$n=8, M=36$      & 1+0+4+8+10+8+4+0+1                                        \\ \hline
$n=9, M=62$      & 1+0+4+9+17+17+11+2+1+0                                    \\ \hline
%$n=10, M=117$     & 1+0+5+7+28+35+28+7+5+0+1                                  \\ \hline
\end{tabular}
\bigskip
\caption{Optimal weight distributions.}
\label{table::optimal weight distributions}
\end{table}

Having in our disposal tables of codes with free points, we can apply
Theorem~\ref{th::main} to construct transmission strategies for an asymmetric channel
with feedback. In the case where the parameters $M(v)$ and $F(v)$ depend on a
codeword weight only, we obtain the following.

\begin{corollary}
Let\/ $M(v)=M_w$ and\/ $F(v)=F_w$ for all\/ $v\in V$ such that the number of
symbols\/~$1$ in $v$ is $w$, i.e., $M(v)$ and $F(v)$ depend on the weight of
the codeword\/ $v$ only. If the conditions
\begin{equation}\label{simplified constaint}
(n_1-w) M_{w+1}\le F_w
\end{equation}
are satisfied for all\/ $w\in[0, n_1-1]$, then the number of transmitted messages
is
\begin{equation}\label{number of transmitted words for special case}
M=\sum\limits_{w=0}^{n_1}\binom{n_1}{w}M_w.
\end{equation}
\end{corollary}

The numbers of messages transmitted by the algorithms constructed based on
Corollary~2 and Theorem~1 are presented in Table~\ref{table::lower bounds for one
feedback}. To compute the values of $M_w$ and~$F_w$ that give optimal answers, we
have used the dynamic programming technique. In the examples presented below, we give
a detailed description of codes obtained using Corollary~2 and Theorem~1 for the
lengths $n=9$ and $n=8$, respectively.

\begin{table}[h]
\renewcommand{\arraystretch}{1.2}% for the vertical padding
\begin{center}
\begin{tabular}{|c|c|c|c|c|c|c|c|c|c|c|}
\hline
$n$  & 5 & 6  & 7  & 8  & 9  & 10  & 11  & 12  & 13    
\\ \hline
$M$ (Corollary~2) & 9 & 16 & 29 & 52 & 96 & 177 & 327 & 607 & 1120 %& 2107 
\\ \hline
$M$ (Theorem~1) & 9 & 16 & 29 & 53 & 97 & $\ge 177$ & $\ge 329$ & $\ge 607$ & $\ge 1120$ %& 2107 
\\ \hline
\end{tabular}
\end{center}
\bigskip

\caption{Number of messages transmitted through an asymmetric
channel with one-time feedback and a single error.}\label{table::lower bounds for one feedback}
\end{table}

We will denote by $(n, M, F, t)_Z$ a nonadaptive code of length $n$ correcting $t$
asymmetric errors and having $M$ codewords and $F$ free points. In the first example
we demonstrate an application of Corollary 2, where the code used after the feedback
depends only on the weight of a codeword transmitted before the feedback. On the
length $n=8$, applying Corollary~2 allows to transmit $52$ messages only. In the
second example we show how one can transmit $53$ messages using Theorem~1. This means
that Corollary~2 does not always give an optimal answer.

\begin{Example}
$n=9$ and $M=96$.

Let $n_1=5$ and $n_2=4$. To vertices that correspond to binary words of weights 0
and~1, we assign a $(4,2,12,1)_Z$ code with two codewords $\{0000, 0011\}$. To
vertices of weights $2$ and $3$, we assign a $(4,3,9,1)_Z$ code with three codewords
$\{0000, 0011, 1100\}$. To weights $4$ and $5$, we assign a $(4,4,4,1)_Z$ code with
four codewords $\{0000, 0011, 1100, 1111\}$.

Let us check the constraints $(n_1-w)M_{w+1}\le F_w$ for $w\in[0, 4]$:
\[
\begin{aligned}
w&=0:\ 5\cdot2\le12,\\ w&=1:\ 4\cdot3\le12,\\ w&=2:\ 3\cdot3\le9,\\ w&=3:\
2\cdot4\le9,\\ w&=4:\ 1\cdot4\le4.
\end{aligned}
\]

Using~\eqref{number of transmitted words for special case}, we compute the number of
transmitted messages:
$$
1\cdot2 +5\cdot2 +10\cdot3 +10\cdot3 +5\cdot4+1\cdot4=96.
$$
\end{Example}

\begin{Example}
$n=8$ and $M=53$.

Let $n_1=6$ and $n_2=2$. To the vertex $111111$ we assign a $(2,2,0,1)_Z$ code $\{00,
11\}$. To the vertices $111000$, $001110$, $010101$, $100011$, $100100$, $010010$,
$001001$, $110000$, $010100$, $001000$, $000010$, and $000001$ we assign a
$(2,0,4,1)_Z$ code with 0 codewords. To all other vertices, we assign a~$(2,1,3,1)_Z$
code with one codeword $\{00\}$. One can easily check that
conditions~\eqref{constraints} of Theorem~\ref{th::main} are fulfilled. For example,
let us check the conditions for the vertex $v=101000$: in total, there are four
vertices from which~$v$ can be reached: $\{111000, 101100, 101010, 101001\}$. The
cardinalities of the corresponding codes are
$$
\begin{gathered}
M(111000)=0,\qquad M(101100)=1,\\ M(101010)=1,\qquad M(101001)=1.
\end{gathered}
$$
The sum of these cardinalities is not greater than $F(101000)=3$; i.e., the
constraint for the vertex $v=101000$ is satisfied. In the same way one can check the
constraints for the other vertices.

The total number of transmitted messages is $2+0+51=53$.
\end{Example}

We present a table with the number of messages that can be transmitted through a
channel with a single asymmetric error and complete feedback. The best results are
obtained in~\cite{cicalese2000optimal}, where transmission of $M=2^m$ messages
requires the length $n=m-1+\lceil \log_2 (m+3)\rceil$. Although we use one-time
feedback only, we can transmit more messages than in~\cite{cicalese2000optimal} for
$n\le 13$, except for the case of $n=7$. For an asymmetric channel with complete
feedback and a single error, one can use an algorithm similar to DADA which allows to
construct codes with the optimal number $M_{\rm ad}$ of transmitted messages.
Cardinalities of these codes are presented in Table~\ref{table::adaptive strategies}.
A detailed description of the construction of such codes will be given in one of
subsequent papers.

\begin{table}[h]
\renewcommand{\arraystretch}{1.2}% for the vertical padding
\begin{center}
\begin{tabular}{|c|c|c|c|c|c|c|c|c|c|c|}
\hline
$n$  & 5 & 6  & 7  & 8  & 9  & 10  & 11  & 12  & 13    \\ \hline
$M$\cite{cicalese2000optimal} & 8 & 16 & 32 & 32 & 64 & 128 & 256 & 512 & 1024 %& 2048 
\\ \hline
$M_{ad}$ & 11 & 20 & 36 & 66 & 121 & 223 & 415 & 774 & 1452 \\ \hline
\end{tabular}
\end{center}
\bigskip
\caption{Code cardinalities for an asymmetric channel with complete feedback
and a single error.}\label{table::adaptive strategies}
\end{table}

\section*{FUNDING}

The research of I.V. Vorobyev was partially supported by the joint grant of the
Russian Foundation for Basic Research and the National Science Foundation of
Bulgaria, project no.~20-51-18002, Russian Foundation for Basic Research, project
no.~20-01-00559, and BMBF-NEWCOM, grant no.~16KIS1005.

The research of K. Deppe was partially supported by BMBF-NEWCOM, grant no.~16KIS1005,
and BMBF-6G-life, grant no.~16KISK002.

The research of A.V. Lebedev and V.S. Lebedev was partially supported by the joint
grant of the Russian Foundation for Basic Research and the National Science
Foundation of Bulgaria, project no.~20-51-18002.


\begin{thebibliography}{18}

\bibitem{renyi61}
R\'enyi, A., On a Problem of Information Theory, \emph{Magyar Tud.\ Akad.\ Mat.\
Kutat\'o Int.\ K\"ozl.}, 1961, vol.~6, pp.~505--516.

\bibitem{B68}
Berlekamp, E.R., Block Coding for the Binary Symmetric Channel with Noiseless,
Delayless Feedback, \emph{Error-Correcting Codes (Proc.\ Conf.\ Conducted by the
Mathematics Research Center, United States Army, at the University of Wisconsin,
Madison, May~6--8, 1968)}, Mann, H.B., Ed., New York: Wiley, 1969, pp.~61--85.

\bibitem{zigangirov1976}
Zigangirov, K.Sh., On the Number of Correctable Errors for Transmission over a Binary
Symmetrical Channel with Feedback, \emph{Probl.\ Peredachi Inf.}, 1976, vol.~12,
no.~2, pp.~3--19 [\emph{Probl.\ Inf.\ Transm.}\ (Engl.\ Transl.), 1976, vol.~12,
no.~2, pp.~85--97]. %\url{http://mi.mathnet.ru/eng/ppi1683}

\bibitem{ulam1991adventures}
Ulam, S.M., \emph{Adventures of a Mathematician}, New York: Scribner, 1976.

\bibitem{P87}
Pelc, A., Solution of Ulam's Problem on Searching with a Lie, \emph{J.~Combin.\
Theory Ser.~A}, 1987, vol.~44, no.~1, pp.~129--140.
%\url{https://doi.org/10.1016/0097-3165(87)90065-3}

\bibitem{G90}
Guzicki, W., Ulam's Searching Game with Two Lies, \emph{J.~Combin.\ Theory Ser.~A},
1990, vol.~54, no.~1, pp.~1--19. %\url{https://doi.org/10.1016/0097-3165(90)90002-E}

\bibitem{D00}
Deppe, C., Solution of Ulam's Searching Game with Three Lies or an Optimal Adaptive
Strategy for Binary Three-Error-Correcting Codes, \emph{Discrete Math.}, 2000,
vol.~224, no.~1--3, pp.~79--98. %\url{https://doi.org/10.1016/S0012-365X(00)00109-6}

\bibitem{D02}
desJardins, D.L., Precise Coding with Noiseless Feedback, \emph{PhD Thesis}, Dept.\
of Mathematics, Univ.\ of California, Berkeley, 2002. Available at
%\url{http://www.desjardins.org/david/thesis/thesis.pdf}

\bibitem{rivest1980coping}
Rivest, R.L., Meyer, A.R., Kleitman, D.J., Winkelmann, K., and Spencer, J., Coping
with Errors in Binary Search Procedures, \emph{J.~Comput.\ System Sci.}, 1980,
vol.~20, no.~3, pp.~396--404. %\url{https://doi.org/10.1016/0022-0000(80)90014-8}

\bibitem{Cicalese13}
Cicalese, F., \emph{Fault-Tolerant Search Algorithms: Reliable Computation with
Unreliable Information}, Berlin: Springer, 2013.

\bibitem{cicalese2000optimal}
Cicalese, F. and Mundici, D., Optimal Coding with One Asymmetric Error: Below the
Sphere Packing Bound, \emph{Computing and Combinatorics (Proc.\ 6th Annu.\ Int.\
Conf.\ COCOON 2000, Sydney, Australia, July 26--28, 2000)}, Du, D.Z., Eades, P.,
Estivill-Castro, V., Lin, X., and Sharma, A., Eds., Lect.\ Notes Comput.\ Sci.,
vol.~1858, Berlin: Springer, 2000, pp.~159--169.
%\url{https://doi.org/10.1007/3-540-44968-X_16}

\bibitem{dumitriu2004halfliar}
Dumitriu, I. and Spencer, J., A Halfliar's Game, \emph{Theoret., Comput Sci.}, 2004,
vol.~313, no.~3, pp.~353--369. %\url{https://doi.org/10.1016/j.tcs.2002.09.001}

\bibitem{spencer2003halflie}
Spencer, J. and Yan, C.H., The Halflie Problem, \emph{J.~Combin.\ Theory Ser.~A},
2003, vol.~103, no.~1, pp.~69--89.
%\url{https://doi.org/10.1016/S0097-3165(03)00068-2}

\bibitem{bassalygo2005}
Bassalygo, L.A., Nonbinary Error-Correcting Codes with One-Time Error-Free Feedback,
\emph{Probl.\ Peredachi Inf.}, 2005, vol.~41, no.~2, pp.~63--67 [\emph{Probl.\ Inf.\
Transm.}\ (Engl.\ Transl.), 2005, vol.~41, no.~2, pp.~125--129].
%\url{https://doi.org/10.1007/s11122-005-0017-3}

\bibitem{dumitriu2005two}
Dumitriu, I. and Spencer, J., The Two-Batch Liar Game over an Arbitrary Channel,
\emph{SIAM J. Discrete Math.}, 2005, vol.~19, no.~4, pp.~1056--1064.
%\url{https://doi.org/10.1137/040617510}

\bibitem{delsarte1981bounds}
Delsarte, P. and Piret, P., Bounds and Constructions for Binary Asymmetric
Error-Correcting Codes, \emph{IEEE Trans.\ Inform.\ Theory}, 1981, vol.~27, no.~1,
pp.~125--128. %\url{https://doi.org/10.1109/TIT.1981.1056290}

\bibitem{klove1981upper}
Kl{\o}ve, T., Upper Bounds on Codes Correcting Asymmetric Errors, \emph{IEEE Trans.\
Inform.\ Theory}, 1981, vol.~27, no.~1, pp.~128--131.
%\url{https://doi.org/10.1109/TIT.1981.1056291}

\bibitem{weber1987new}
Weber, J., de Vroedt, C., and Boekee, D., New Upper Bounds on the Size of Codes
Correcting Asymmetric Errors, \emph{IEEE Trans.\ Inform.\ Theory}, 1987, vol.~33,
no.~3, pp.~434--437. %\url{https://doi.org/10.1109/TIT.1987.1057301}

\end{thebibliography}
\end{document}